\pdfoutput=1 

\newif\ifconference
\conferencefalse
\ifconference
\documentclass[twoside,leqno,11pt]{article}
\usepackage{ltexpprt}
\usepackage{hyperref}
\else
\documentclass[11pt]{article}
\usepackage[hyperindex,breaklinks,bookmarks]{hyperref}
\usepackage{amsthm}
\fi
\usepackage[utf8]{inputenc}
\usepackage[english]{babel}
\usepackage{amsmath,amsfonts}
\usepackage{amssymb}
\usepackage[showonlyrefs]{mathtools}
\usepackage{moresize}
\usepackage[a4paper, portrait, left=1in, right=1in, top=1in, bottom=1in]{geometry}
\usepackage[numbers,sectionbib,longnamesfirst]{natbib}
\hypersetup{bookmarksdepth=2}
\usepackage[shortlabels]{enumitem}
\usepackage[toc]{appendix}
\usepackage{microtype}
\usepackage[ruled,linesnumbered]{algorithm2e}
\usepackage{cleveref}
\usepackage{xcolor}

\DeclareMathOperator{\ess}{Ess}
\newcommand{\Pp}{P}

\SetKwFor{RepTimes}{repeat}{times}{end}

\ifconference \else
\newtheorem{theorem}{Theorem}
\newtheorem{corollary}[theorem]{Corollary}
\newtheorem{definition}[theorem]{Definition}
\newtheorem{lemma}[theorem]{Lemma}
\newtheorem{observation}[theorem]{Observation}

\theoremstyle{remark}
\newtheorem*{remark}{Remark}
\fi

\newcommand{\BE}{\mathbb E}

\newcommand{\BI}{\mathbb I}

\newcommand{\BP}{\mathbb P}

\newcommand{\BR}{\mathbb R}

\newcommand{\Var}{\operatorname{Var}}

\date{}
\ifconference
\title{\Large Estimating the Effective Support Size in Constant Query Complexity}
\else
\title{Estimating the Effective Support Size in Constant Query Complexity}
\fi

\ifconference
\author{Shyam Narayanan \thanks{MIT, \texttt{shyamsn@mit.edu}, supported by the NSF GRFP Fellowship, and the NSF TRIPODS Program (award DMS-2022448). This work was partially done while he was visiting BARC, University of Copenhagen.}\and Jakub Tětek\thanks{BARC, Univ. of Copenhagen, \texttt{j.tetek@gmail.com}, supported by the VILLUM Foundation grant 16582.}}
\else
\author{
    Shyam Narayanan\thanks{Shyam Narayanan is supported by the NSF GRFP Fellowship, and the NSF TRIPODS Program (award DMS-2022448)}\\
    \texttt{shyamsn@mit.edu} \\ 
    MIT
    \and
    Jakub Tětek\thanks{Jakub T\v{e}tek is supported by the VILLUM Foundation grant 16582.}\\
    \texttt{j.tetek@gmail.com}\\
    BARC, Univ. of Copenhagen
}
\fi

\begin{document}
\maketitle

\ifconference
\fancyfoot[R]{\scriptsize{Copyright \textcopyright\ 2023 by SIAM\\
Unauthorized reproduction of this article is prohibited}}
\else
\thispagestyle{empty}
\fi

\begin{abstract}
\ifconference \small\baselineskip=9pt  \fi
    Estimating the support size of a distribution is a well-studied problem in statistics. Motivated by the fact that this problem is highly non-robust (as small perturbations in the distributions can drastically affect the support size) and thus hard to estimate, Goldreich [ECCC 2019] studied the query complexity of estimating the $\epsilon$-\emph{effective support size} $\ess_\epsilon$ of a distribution $\Pp$, which is equal to the smallest support size of a distribution that is $\epsilon$-far in total variation distance from $\Pp$.
    
    In his paper, he shows an algorithm in the dual access setting (where we may both receive random samples and query the sampling probability $p(x)$ for any $x$) for a bicriteria approximation, giving an answer in $[\ess_{(1+\beta)\epsilon},(1+\gamma) \ess_{\epsilon}]$ for some values $\beta, \gamma > 0$. However, his algorithm has either super-constant query complexity in the support size or super-constant approximation ratio $1+\gamma = \omega(1)$. He then asked if this is necessary, or if it is possible to get a constant-factor approximation in a number of queries independent of the support size.
  %  
%    ,yet this problem is known to be highly non-robust as small perturbations in the distributions can drastically affect the support size.
%    In this paper, we study the problem of estimating the \emph{effective support size} of a distribution $\Pp$, which asks to compute the smallest support size of a distribution that is close to $\Pp$. This notion, first studied by \cite{}, is significantly more robust as we consider any distribution that is close to $\Pp$.
%    
%
    
    We answer his question by showing that not only is complexity independent of $n$ possible for $\gamma>0$, but also for $\gamma=0$, that is, that the bicriteria relaxation is not necessary. Specifically, we show an algorithm with query complexity $O(\frac{1}{\beta^3 \epsilon^3})$. That is, for any $0 < \epsilon, \beta < 1$, we output in this complexity a number $\tilde{n} \in [\ess_{(1+\beta)\epsilon},\ess_\epsilon]$. 
    %such that $\Pp$ is $(1+\beta)\epsilon$-far in total variation distance from any distribution of support size $\tilde{n}$, yet $\Pp$ is $(1+\beta) \cdot \epsilon$-close in total variation distance to some distribution of suppport size at most $\tilde{n}$.
    We also show that it is possible to solve the approximate version with approximation ratio $1+\gamma$ in complexity $O\left(\frac{1}{\beta^2 \epsilon} + \frac{1}{\beta \epsilon \gamma^2}\right)$. 
%
%    We consider the model where we are given oracle access to \emph{probability-revealing samples}~\cite{onaksun}, where each sample is a pair $(i, p_i)$ where $i \sim \Pp$ and $p_i = \BP_{x \sim \Pp}(x=i)$ represents the probability of actually sampling $i$.
%    We answer a conjecture of Goldreich~\cite{} by presenting an algorithm that estimates the effective support size with sample complexity \emph{independent} of the overall support size of the distribution $\Pp$, or of the domain of elements $\Pp$ could come from. 
%    Specifically, for any $0 < \epsilon, \beta, \gamma < 1$ and given $\tilde{O}\left(\frac{1}{\epsilon \cdot \beta^2 \cdot \gamma^2}\right)$\todo{Correct sample complexity} random samples $(i, p_i)$, we can output a number $\tilde{n}$ such that $\Pp$ is $\epsilon$-far in total variation distance from any distribution of support size $\tilde{n}$, yet $\Pp$ is $(1+\beta) \cdot \epsilon$-close in total variation distance to some distribution of suppport size at most $(1+\gamma) \cdot \tilde{n}$.
    Our algorithm is very simple, and has $4$ short lines of pseudocode.
\end{abstract}

\newpage
\pagenumbering{arabic}

\section{Introduction} 
%\todo{$\Pp$ for distribution, $p(x)$ for probability}

Estimating the support size of a distribution is one of the most fundamental problems in statistics and has been studied over many decades, starting with the paper of \citet{fisher1943relation} in 1943. %The problem has been approached from many angles and has been studied in various settings. One such work is the paper by \citet{goldreich2019complexity}, upon which we improve in this paper. Specifically, in his paper, Goldreich considered the dual access setting. This is a commonly studied setting in distribution testing, see \cite{canonne2014testing}.
%One setting common in distribution testing is the dual access setting, introduced by \citet{canonne2014testing}. This is the setting that we consider in this paper.
%
Estimating the support size in full generality is, however, impossible. This is because any distribution is infinitesimally close to a distribution with arbitrarily large support. One common approach is to assume a lower bound on the elements' probabilities \cite{raskhodnikova2009strong,valiant2011estimating,valiant2013estimating,wu2019chebyshev}. This assumption is, however, not always reasonable in practice. This motivated \citet{goldreich2019complexity} to study algorithms for estimating a relaxed quantity known as the \emph{effective support size}.
The $\epsilon$-effective support size (abbreviated as $\ess_\epsilon$) of a distribution $\Pp$ is defined as the smallest $n$ such that there exists a distribution $\Pp'$ supported on $n$ elements that is $\epsilon$-far in total variation distance, that is $\|\Pp - \Pp'\|_{TV} = \epsilon$. This problem is also too non-robust to be estimable in sublinear complexity. However, it leads to a natural bicriteria approximation: we may ask to find a value in $[\ess_{(1+\beta)\epsilon}, (1+\gamma)\ess_{\epsilon}]$ for some $\gamma, \beta>0$. 

Support size estimation fits in the general subfield of distribution testing, where the goal is to test or learn properties of a distribution from samples or queries to the distribution. Various settings have been considered in the context of distribution testing. One common setting is the dual access setting \cite{batu2005complexity,guha2006streaming,canonne2014testing}, where in addition to sampling access to the distribution, we may ask for the sampling probability $p(x)$ of any item $x$. 
Goldreich \cite{goldreich2019complexity} studied the effective support size estimation in exactly this setting, and this is also the setting we use in this paper. In \cite{goldreich2019complexity}, Goldreich gave algorithms that either have super-constant approximation ratio of $O\left(\log^{(k)} (n/\epsilon)\right)$ for any fixed constant $k$, or $\gamma=0$ and query complexity $O(\log^*(n/\epsilon)/\text{poly}(\epsilon \beta))$\footnote{We use $\log^{(k)}(x)$ to denote the $k$th iterated logarithm of $x$, i.e., $\log^{(1)}(x) := \log x$ and for $k \ge 2$, $\log^{(k)}(x) := \log (\log^{(k-1)}(x))$. We use $\log^*(x)$ to denote the smallest nonnegative integer $t$ such that $\log^{(t)}(x) \le 1$.}. In this paper, we show that it is possible to get the best of both worlds: query complexity independent of the support size $n$ in complexity $O(\text{poly}(1/(\epsilon \beta)))$, and having $\gamma = 0$. That is, we show that the bicriteria approximation is not necessary, and that the relaxation of the problem in terms of $\beta>0$ is sufficient for the problem to be efficiently solvable. This answers positively the following question posed by Goldreich. Specifically, all of the following questions have a positive answer:

%The problem was been studied by \citet{} in 1943 in the context of estimating the nyumber of estimating . A more recent line of research lead to an algorithm that returns an estimate by using $O(\ell / \log \ell)$ samples where $1/\ell$ is a lower bound on the probabilities, and this is known to be optimal if the algorithm only has access to the samples. However, in practice, some of the probabilities may be very small. Motivated by this, \citet{} introduced a
%
%\citet{} introduced a model, where we have an oracle access to the distribution. 
%
%
%Lots of study in the distribution testing setting in: sampling \cite{valiantvaliant, wuyang, learning}, dual \cite{canonnerubinfeld, onaksun, }.

\begin{quote}
\textbf{Open Problem 1.10} from \cite{goldreich2019complexity}, \textbf{Open Problem 100} from \cite{sublinear} \textit{(approximators of the effective support size with performance guarantees that are oblivious of the distribution)}:
For a constant $\beta>0$, does there exist an algorithm that, on input $\epsilon>0$ and oracle access to $\Pp$, uses $s(\epsilon)$ queries and outputs an $f(\epsilon)$-factor approximation of the $[\epsilon, (1+\beta) \cdot \epsilon]$-effective support size of $\Pp$, where $s$ and $f$ are functions of $\epsilon$ only? If so, can both functions be polynomials in $1 / \epsilon$? And, if so, can we have $s(\epsilon)=\operatorname{poly}(1 / \epsilon)$ and $f=1$?\footnote{We remark that we have slightly rephrased their problem: in this paper, we set $1+\beta$ to be what they have set as $\beta$.}
\end{quote}

Specifically, we give an algorithm for $1+\gamma$-approximate $[\epsilon, (1+\beta)\epsilon]$-effective support size in time $O\left(\frac{1}{\beta^2 \epsilon} + \frac{1}{\beta \epsilon \gamma^2}\right)$. How do we decrease the $\gamma$ to $0$? Goldreich proved that one may decrease the $\gamma$ to $0$, at the cost of an increase in $\beta$ by a factor of $\gamma/\epsilon$ (Observation~\ref{obs:tradeoff_error_parameters}). This means we may compute a $\beta \epsilon/2$-approximate estimate of the $[\epsilon, (1+\beta/2)\epsilon]$ effective support size, which will also be a 1-approximation of the $[\epsilon, (1+\beta)\epsilon]$-effective support size, as desired. This results in an algorithm with complexity $O\left(\frac{1}{\beta^3 \epsilon^3}\right)$ for the unicriterion approximation version of the problem.

\subsection{Our techniques.}
If our distribution is uniform, it would be natural to sample an item $y$ and return $1/p(y)$, which would be equal to the universe size. It is easy to show that this in fact gives an unbiased estimate for general distributions:
\[
\BE\left[\frac{1}{p(y)}\right] = \sum_{y \in U} p(y) \frac{1}{p(y)} = \sum_{y \in U} 1 = |U|.
\]
This simple estimator has in fact been used to estimate support size, such as in ~\cite{canonne2014testing, onak2018probability}. Our estimator uses this observation as a starting point, and bears resemblance to~\cite{canonne2014testing, onak2018probability}.

The $\epsilon$-effective support size corresponds to ignoring the smallest probability items totaling $\epsilon$ probability mass. Let us therefore order the universe in order of increasing probabilities. We may then modify the above estimator as follows. We generate a sufficiently large random sample of items drawn from $\Pp$ and compute the $(1+\beta/2)\epsilon$-quantile $x$ of the samples with respect to the order, and define $p = \Pp(x)$. If we sampled enough items, it should hold that $\BP_{X \sim \Pp} (X < x) = \epsilon^*$, for $\epsilon^* \approx (1+\beta/2)\epsilon$. If all probabilities are distinct, we may then use $\BI[p(X) \geq p(x)]/p(X),$ where $X \sim \Pp$, as an unbiased estimate\footnote{The unbiasedness also follows from the fact that this is a special case of the Hansen-Hurwitz estimator \cite{hansen1943theory}.} of $\ess_{\epsilon^*}$:
\[
\BE\left[\frac{\BI[p(X) \geq p(x)]}{p(X)}\right] = \sum_{y \in U} p(y) \frac{\BI[p(y) \geq p(x)]}{p(y)} = \sum_{y \in U} \BI[p(y) \geq p(x)] = \ess_{\epsilon^*}.
\]
Here, we use $\BI$ to denote the indicator random variable for an event. The final equality holds from a known observation (see Observation \ref{obs:equivalent_definition}) that if the $n$ heaviest elements in $\Pp$ have total probability $1-\epsilon$, then $\ess_{\epsilon} = n$.

We now bound the variance.
In some of the cases, we may use a straightforward analysis that we will now describe; we briefly describe the final and most difficult case at the end of this subsection.
%in a straightforward way. In the rest of the proof, we then argue that this bound is sufficient. 

Specifically, we use the fact that for a random variable $X$ with $X \geq 0$ almost surely, $\Var[X] \leq \BE[X] \sup[X]$, where $\sup[X]$ represents the maximum value $X$ may take. Because an indicator variable is at most $1$ and $p(X) \ge p(x) = p$ whenever the indicator is true, this gives us
\[
\Var\left[\frac{\BI[p(X) \geq p(x)]}{p(X)}\right] \leq \BE\left[\frac{\BI[p(X) \geq p(x)]}{p(X)}\right] \cdot \frac{1}{p} = \frac{\ess_{\epsilon^*}}{p}.
\]

If it were the case that $\ess_{\epsilon^*} \geq \frac{\epsilon^*\beta}{100 p}$, this would be sufficient, as we could use this to upper-bound the variance by $\leq \frac{100}{\epsilon^* \beta} \ess_{\epsilon^*}^2$ which directly leads by Chebyshev's inequality to an algorithm with complexity independent of $n$.

The difficult case is thus when $\ess_{\epsilon^*} < \frac{\epsilon^*\beta}{100 p}$. The basic idea is that in this case, we can show that $\ess_{(1-\beta/4)\epsilon^*}$ is significantly larger than $\ess_{\epsilon^*}$, meaning that the interval $[\ess_{\epsilon^*},\ess_{(1-\beta/4)\epsilon^*}]$ is large. As we are assuming $\epsilon^* \approx (1+\beta/2)\epsilon$, any answer in this range is valid. Intuitively speaking, the fact that the range of valid outputs is large then makes the problem easier.

The items that are counted in $\ess_{(1-\beta/4)\epsilon^*}$ but not in $\ess_{\epsilon^*}$ have $\beta\epsilon^*/4 \geq \beta\epsilon/4$ probability mass, but each has a probability of being sampled at most $p$. There are therefore at least $\beta\epsilon/(4p)$ of them. Hence, $\ess_{(1-\beta/4)\epsilon^*} \geq \beta\epsilon/(4p)$ and $p \geq \beta \epsilon/(4\ess_{(1-\beta/4)\epsilon^*})$. 
This also implies that $\ess_{(1-\beta/4)\epsilon^*} \ge 2 \ess_{\epsilon^*}$, since we know $\ess_{\epsilon^*} < \frac{\epsilon^* \beta}{100 p},$ and $\epsilon^* \approx (1+\beta/2) \epsilon$.
We now argue that we return a value $\leq \ess_{(1-\beta/4)\epsilon^*} \leq \ess_{\epsilon}$. Since each estimate is unbiased with variance at most $\ess_{\epsilon^*}/p$, by the Chebyshev inequality, if we average this estimate over $t$ samples, with high probability we return a value that is at most
\[
\ess_{\epsilon^*} + O\left(\frac{1}{\sqrt{t}} \cdot \sqrt{\ess_{\epsilon^*}/p}\right) \leq \ess_{\epsilon^*} + O\left(\frac{1}{\sqrt{t}} \cdot \sqrt{\ess_{\epsilon^*}\ess_{(1-\beta/4)\epsilon^*}/(\epsilon \beta)}\right) \leq \ess_{(1-\beta/4)\epsilon^*} \leq \ess_{\epsilon}.
\]
Above, the first inequality holds by Chebyshev's inequality, the second inequality holds by our assumption $\ess_{(1-\beta/4)\epsilon^*} \geq \beta\epsilon/(4p)$, and the last inequality holds by the assumption $\epsilon^* \approx (1+\beta/2)\epsilon$. The third inequality holds as long as $t \ge \frac{C}{\epsilon \beta}$ for some large constant $C$, since $\ess_{\epsilon^*} \le \frac{1}{2} \cdot \ess_{(1-\beta/4)\epsilon^*}$, 
Importantly, $t$ only needs to depend on $\beta$ and $\epsilon$, not on $n$.
%where the second inequality holds for $t$ large enough but independent of $n$, because the first term is $\leq \ess_{(1-\beta/4)\epsilon^*}/25$ \footnote{Note that we are assuming $\ess_{\epsilon^*} < \epsilon^* \beta/(100 p)$, whereas we argued $\ess_{(1-\beta/4)\epsilon^*} \geq \beta\epsilon/(4p)$}, and the second term is $O(\ess_{(1-\beta/4)\epsilon^*}/(t \sqrt{\epsilon \beta}))$, so we can make the second term, say, $\leq \ess_{(1-\beta/4)\epsilon^*}/2$, making the inequality hold. The last inequality holds by the assumption $\epsilon^* \approx (1+\beta/2)\epsilon$.

It remains to prove that we return a value that is at least
$\ess_{(1+\beta)\epsilon}$.
%$\geq \ess_{(1+\beta/7)\epsilon^*} \geq \ess_{(1+\beta)\epsilon}$ (again, we are using that $\epsilon^* \approx (1+\beta/2)\epsilon$ in the second inequality; the exact value of the constant $1/7$ is not important, we use $1/7$ in the analysis below). 
Unfortunately, the variance of our estimator can be arbitrarily large (if one of the probabilities is extremely small), and we thus cannot use Chebyshev's inequality to prove that the returned value will be close to the expectation, and thus not too small. We get around this issue by show a different random variable that is stochastically dominated\footnote{We define stochastic domination in \Cref{sec:prelims}.} by our estimator, and whose variance is small enough and expectation large enough for this argument to work. Since our estimator stochastically dominates this random variable, it is also not too small with good probability.

\subsection{Related work.}
%\jakub{basically merge of the lit review from your supp size estim paper, and the paper by Goldreich?}

The problem of support size estimation has been studied over many decades. To the best of our knowledge, the problem was first considered in 1943 under parametric assumptions by \citet{fisher1943relation}. Under slightly different assumptions, the problem was then considered in 1953 by \citet{good1953population}. A large number of works have since followed (see \cite{gandolfi2004nonparametric} for a survey). However, no approach with formal guarantees without parametric assumptions was known until the study of this problem in the context of distribution testing.

%Support size estimation fits in the general framework of distribution testing, where the goal is to learn or test properties of a distribution from samples or queries. 
Distribution testing has also enjoyed a long line of research over the past few decades (see~\citet{canonnesurvey} for a survey).
The study of the support size estimation problem in the context of distribution testing started more recently, with \cite{raskhodnikova2009strong} in 2009. Perhaps the most common parametrization of this problem in distribution testing is by the smallest probability of any item. That is, one assumes that $1/n \leq p(x)$ for any item $x$ in the universe, and $n$ is now no longer the universe size. In this setting, a line of research \cite{raskhodnikova2009strong,valiant2011estimating,valiant2013estimating,wu2019chebyshev} lead to an algorithm with complexity $O(\frac{n}{\log n} \log^2 (1/\epsilon))$ to estimate the support up to additive error $\epsilon \cdot n$ in the setting when we have sampling access to the distribution.

There are several settings that are commonly studied in distribution testing. Among them are the dual setting, notably systematically studied by \citet{canonne2014testing}, and probability-revealing samples defined by \citet{onak2018probability}. The dual setting has also been considered prior to \cite{canonne2014testing} in \cite{batu2005complexity,guha2006streaming}. The dual setting assumes that we may ask for the sampling probability of an item. In the probability-revealing samples setting, we get with each sampled item, its sampling probability. The difference is that we may ask the dual oracle for probabilities even of items that have not been sampled. 
A related setting is the ``learning-based'' distribution framework~\cite{eden2021learning}, which is similar to the probability-revealing samples setting except with each sampled item, we only receive an $O(1)$-approximation to the sampling probability rather than the exact sampling probability.
In the dual and probability-revealing samples settings, it is possible to get in time $O(1/\epsilon^2)$ an additive $\pm \epsilon n$ approximation, and in the learning-based setting, it is possible to get the same approximation in time $O(n^{1-\Theta(1/\log \epsilon^{-1})} \cdot \log \epsilon^{-1})$~\cite{eden2021learning}. In all of these settings, we again choose $n$ such that $1/n \le p(x)$ for all $x$ in the universe, which means $n$ can be much larger than the universe size \cite{canonne2014testing,onak2018probability, eden2021learning}. This may, however, be a very poor approximation, if some sampling probabilities are very small. This motivates the notion of effective support size, as this is known to be optimal \cite{canonne2014testing,onak2018probability} and relative approximation is thus not possible in complexity independent of $n$.

%One can easily show that for the problem of estimating the effective support size, these two settings are equivalent, as if the true support size is huge, then asking for the probability of any item that we have not sampled is unlikely to give us any useful information.

While effective support size was first defined by \citet{blais2017ess} and also studied in \citet{StewartDC18}, the specific problem of \emph{estimating} effective support size was first studied later by \citet{goldreich2019complexity}. The main motivation for this relaxation of the problem is that it is possible to get a relative approximation to the effective support size, even if there are no promises on the minimum probability. Specifically, Goldreich shows that for any $\epsilon>0$ and any fixed $\beta > 0$, it is possible to get a $(1+\gamma)$-approximation to the $[\epsilon, (1+\beta)\epsilon]$-effective support size, in complexity $s$, for:
%an $O(\log \log \cdots \log (n/\epsilon))$ using $O(1/\text{poly}(\beta \epsilon))$ samples, where $\log \log \cdots \log$ represents any fixed number of logarithms, as well as a $O(1)$-approximation in $O(\log^* (n/\epsilon)/\text{poly}(\beta \epsilon))$-approximation.
%
\begin{enumerate}
\setlength\itemsep{0.0em}
    \item $s=O(1 / \epsilon)$ and $1+\gamma=O\left(\epsilon^{-1} \log (n / \epsilon)\right)$,
    \item $s=\widetilde{O}(1 / \epsilon)$ and $1+\gamma=O(\log (n / \epsilon))$.
    \item For any constants $t, k \in \mathbb{N}$, it holds that $s=\widetilde{O}\left(t / \epsilon^{1+\frac{1}{k}}\right)$ and $1+\gamma=\widetilde{O}(\log^{(t)}(n / \epsilon))$.
    %, where $\log^{(t)}$ denotes $t$ iterated logarithms.
    \item For any constant $k \in \mathbb{N}$, it holds that $s=\widetilde{O}\left(\log ^{*}(n / \epsilon) / \epsilon^{1+\frac{1}{k}}\right)$ in expectation and $\gamma=\beta$.
\end{enumerate}

\paragraph{Simplicity:} To our knowledge, the only work to study estimating effective support size is that of \citet{goldreich2019complexity}. We note that our algorithm not only achieves a better query complexity but is also substantially simpler and shorter, both in terms of algorithm description and analysis.

\section{Preliminaries} \label{sec:prelims}
%\begin{itemize}
%   % \item say something about TV distance? (maybe include as footnote in intro)
%    \item say something about the model
%    \item formally define the problem? (may be better to be more formal above by an $\epsilon$ and then not have to re-define?)
%    %\item stochastic domination
%\end{itemize}

\subsection{Effective support size and its properties.}

%\jakub{we could put the ess definition here}

The effective support size of a distribution $\Pp$ is defined as follows.
\begin{definition}[Definition 1.1 from \cite{goldreich2019complexity}]
The $\epsilon$-effective support size $\ess_\epsilon(\Pp)$ is defined as the smallest $n$ such that $\Pp$ is $\epsilon$-close in total variation distance to some distribution $\Pp'$ whose support has size $n$.
\end{definition}

As \citet[Proposition 1.6]{goldreich2019complexity} proves, it is not possible to efficiently estimate the $\epsilon$-effective support size. Instead, one has to use a relaxation of this notion.% We give a slightly different definition, so that $\tilde{n}$ being a $1$-approximate effective $[\epsilon_1,\epsilon_2]$-support size is the same as it being an $[\epsilon_1,\epsilon_2]$-support size.

\begin{definition}[Definition 1.2 from \cite{goldreich2019complexity}]
A value $\tilde{n}$ is a $(1+\gamma)$-approximate effective $[\epsilon_1,\epsilon_2]$-support size if $n \in [\ess_{\epsilon_2}, (1+\gamma)\ess_{\epsilon_1}]$.
\end{definition}

We prove that, in fact, one does need the error parameter $\gamma$ in the sense that the problem is efficiently solvable even for $\gamma=0$.

We now state two observations of \citet{goldreich2019complexity} that we will need. The first says that $\epsilon$-effective support size is equal to the support size after removing the least likely elements with a total mass of $\epsilon$. The second one says that we may decrease $\gamma$ to $0$ at the cost of an increase in $\beta$ by a factor of $\gamma/\epsilon$.

\begin{observation}[Observation 1.4 in \cite{goldreich2019complexity}]\label{obs:equivalent_definition}
If $\Pp$ has $\epsilon$-effective support size $n$, then $\Pp$ is $\epsilon$-close to a distribution that has support that consists of the $n$ heaviest elements in $\Pp$, with ties broken arbitrarily.
\end{observation}

\begin{observation}[Observation 1.5 in \cite{goldreich2019complexity}] \label{obs:tradeoff_error_parameters}
If a random variable $X$ is a $(1+\gamma)$-factor approximation of the $\left[\epsilon_{1}, \epsilon_{2}\right]$-effective support size of $\Pp$, then $X / (1+\gamma)$ is an $\left[\epsilon_{1}, \epsilon_{2}+\gamma/(1+\gamma)\right]$-effective support size of $\Pp$. In particular, for $\gamma=\beta \epsilon$, we have $\gamma/(1+\gamma)<\beta \epsilon$. Therefore, if a random variable $X$ is a $(1+\beta \epsilon)$-factor approximation of the $\left[\epsilon, (1+\beta)\epsilon\right]$-effective support size of $\Pp$, then $X/(1+\gamma)$ is an $[\epsilon, (1+2\beta) \epsilon]$-effective support size of $\Pp$.
\end{observation}

\subsection{Distribution testing settings.}
%\jakub{define the two relevant settings. Maybe explain why they are equivalent for us.}
%The dual access model was first systematically studied by \citet{canonne2014testing}, but has been used before in \cite{batu2005complexity,guha2006streaming}. 
The model for a distribution $\Pp$ on a universe $U$ is defined as a pair of oracles $(\text{SAMP}_{\Pp}, \text{EVAL}_{\Pp})$ which are in turn defined as follows. Upon being queried, $\text{SAMP}_{\Pp}$ returns a sample from $\Pp$, independent from all previously returned samples. $\text{EVAL}_{\Pp}(x)$ for $x \in U$ returns the probability $p(x)$ of $x$ being sampled from $\Pp$.

The probability-revealing samples model was defined by \citet{onak2018probability}. For a distribution $\Pp$, we define a probability-revealing oracle $REV_{\Pp}$ as an oracle that returns $(x, p(x))$ for $x \sim \Pp$, independently of all previous calls of the oracle.
The difference between these two settings that one may also use the $\text{EVAL}$ oracle on items that have not been sampled in the dual access model, but not in the probability-revealing samples model.
Hence, the dual access model in general is more powerful than the probability-revealing samples model.
Our algorithm in fact will only query probabilities for elements that have already been sampled, so it works both in the dual access model and the probability-revealing samples model.

We also briefly remark that, while not actually relevant for our upper bounds, for effective support size estimation, the query complexities in these two models are in fact equivalent.
For symmetric properties, one may assume that we use $\text{EVAL}$ either on sampled items, or on items selected uniformly at random from the not-yet-seen part of the support.\footnote{This may be argued roughly as follows: we take a uniformly random permutation $\pi$ of the universe. By symmetry, this does not affect correctness. At the same time, no matter the distribution of $x$ that the algorithm queries, we have $\text{EVAL}_{\pi^{-1}(\Pp)}(x) = \text{EVAL}_{\Pp}(\pi^{-1}(x))$, but $\pi^{-1}(x)$ has conditional distribution of being uniform on the not-yet-sampled items. This holds even for adaptive queries, as after each adaptive query we have no information on $\pi$ outside of the elements we sampled/queried.} %\jakub{polish this footnote}
In addition, if we want the query complexity to have no dependence on the universe size $|U|$, sampling uniformly from the not-yet-seen part of the support is useless because one can make $|U|$ arbitrarily large by adding elements of probability $0$, and sampling uniformly means that with overwhelming probability we will only see elements with probability $0$.
%This allows us to see that an algorithm for effective support size estimation with complexity independent of the universe size in the dual setting implies such algorithm in the probability-revealing setting. The reduction works as follows. We extend the support by a large number $k$ of items with very small probability $q$. This number pick the number to be so large, that by uniform sampling, we do not expect to see even one item other than these added items. We pick the probabilities so small, that this does not affect the effective support size. If we now have an algorithm in the dual setting, whenever it queries a not-yet-sampled item, we give it $q$, while we multiply all the other probabilities by a factor of $1-k q$. Since, with high constant probability, the algorithm would only receive value $q$ in the dual setting, the algorithm will, with high constant probability, output the same value it would in the dual setting while we have removed the need for the queries on not-sampled elements.\jakub{polish this paragraph. Do we even want it?}

\subsection{Notions from probability theory.}
First, we note that we use $\Pp$ to denote a distribution, and $p(x)$ to denote the probability of sampling $x$ from $\Pp$. We also use $\BI$ to denote an indicator random variable. In other words, for an event $E$, $\BI[E] = 1$ if $E$ occurs and $\BI[E] = 0$ otherwise.

For a real-valued random variable $X$, we define $\sup[X] = \sup_{t \in \BR} \{t: \BP(X \ge t) > 0\}$: $\sup[X]$ roughly represents the largest real number that $X$ may take. If $\BP(X \ge t) > 0$ for all $t \in \BR$, then $\sup[X] = \infty$. In addition, if $X$ is conditioned on some variable $Y$, we can define $\sup[X|Y = y] = \sup_{t \in \BR} \{t: \BP(X \ge t|Y = y) > 0\}$.

In this paper, we need some common notions from probability theory. Two note-worthy ones are that of \textit{total variation distance} and \textit{stochastic domination}. The total variation distance of two distributions $\Pp_1, \Pp_2$ supported on $U$ can be for finite $U$ written as
\[
\|\Pp_1 - \Pp_2\|_{TV} = \frac{1}{2}\sum_{x \in U} |p_1(x) - p_2(x)|
\]

We say that a real random variable $X_1$ stochastically dominates a real random variable $X_2$ if there exists a coupling $(X_1',X_2')$ such that $X_1' \geq X_2'$ almost surely. We also use an equivalent definition, which states $X_1$ stochastically dominates $X_2$, iff for any value $\phi$, we have $\BP[X_1 \geq \phi] \geq \BP[X_2 \geq \phi]$.

\section{Effective support size estimation}
We assume an arbitrary total ordering $\leq$ on the support. This may be assumed WLOG and without seeing any samples; for instance, each element will have some number or categorical label associated with it, so we may set $\leq$ as the natural lexicographic order on the labels.

We then define a total ordering $\prec$ such that $x_1 \prec x_2$ if $p(x_1) < p(x_2)$ or if $p(x_1) = p(x_2)$ and $x_1 < x_2$. (We also define $x_1 \preccurlyeq x_2$ to mean $x_1 \prec x_2$ or $x_1 = x_2$).
We define the $\epsilon$-\emph{quantile} of a distribution $\Pp$ over a universe $U$ to be the $x_\epsilon \in U$ that is smallest w.r.t. $\prec$ such that $\BP_{X \sim \Pp}(X \preccurlyeq x_\epsilon) > \epsilon$.
One can verify (using Observation~\ref{obs:equivalent_definition}) that $\ess_\epsilon$ equals the number of elements that are $\succcurlyeq x_\epsilon$ if $x_\epsilon$ is the $\epsilon$-quantile.
In addition, if given a sample $R$ from the distribution, we define the $\epsilon$-quantile of $R$ w.r.t. $\prec$ to be the smallest $x \in U$ (under the $\prec$ ordering) such that $\#\{X \in R: X \preccurlyeq x\} > \epsilon \cdot |R|$.

%Let $\epsilon^*$ be such that $x$ is the $\epsilon^*$ quantile of $\Pp$ with respect to $\prec$. Let $p = p(x)$ and $p_\epsilon = p(x_\epsilon)$ where $x_\epsilon$ is the $\epsilon$-quantile of $\Pp$.
Finally, for any $0 < \epsilon < 1$, define $x_\epsilon$ as the $\epsilon$-quantile of $\Pp$, and $p_\epsilon := p(x_\epsilon)$.

Given this, we can now describe our algorithm, described in \Cref{alg:main}. Indeed, our algorithm is very simple and only requires $4$ lines of pseudocode description. We will assume WLOG that $\beta \le 0.2$ and $\gamma \le 0.2$ throughout the analysis. In addition, we assume WLOG that $(1+\beta) \cdot \epsilon < 1$, as if $(1+\beta) \cdot \epsilon \ge 1$, then any distribution with support $1$ (i.e., a point mass on any element) has total variation distance at most $1 \le (1+\beta) \cdot \epsilon$ from $\Pp$, so we may output $1$ as our estimate of the effective support size.
\begin{algorithm}
$R \leftarrow $ sample of size $\frac{180}{\beta^2 \epsilon}$\\
$x \leftarrow$ $(1+\beta/2) \epsilon$-quantile of $R$ w.r.t.~$\prec$\\
$y_1, \dots, y_t \leftarrow$ sample of size $t = \frac{500}{\epsilon \beta \gamma^2}$\\
\Return{$(1+\gamma/2) S$, where $S = \frac{1}{t} \sum_{i=1}^{t} \frac{\mathbb{I}[y_i \succcurlyeq x]}{p(y_i)} $}\\

\caption{Get a $(\gamma,\beta)$-approximate estimate of the $\epsilon$-effective support size.}
\label{alg:main}
\end{algorithm}

First, we show that the $x$ created in \Cref{alg:main} is an approximate $\epsilon$ quantile of $\Pp$.

\begin{lemma} \label{lem:first_2_lines}
    Let $x$ represent the output of the second line of \Cref{alg:main}. With probability at least $\frac{9}{10},$ there exists $\epsilon^* \in [(1+\beta/4) \epsilon, (1 + 3\beta/4) \epsilon]$ such that $x$ is the $\epsilon^*$ quantile of $\Pp$.
\end{lemma}

\begin{remark}
    We say ``there exists $\epsilon^*$'' as the choice of $\epsilon^*$ may not be unique. For instance, if $\Pp$ were a point mass on a single element $x$, then $x$ is the $\epsilon$ quantile for all $0 < \epsilon < 1$.
\end{remark}

\begin{proof}
    Let $x_{(1+\beta/4)\epsilon}$ represent the $(1+\beta/4)\epsilon$ quantile of $\Pp$. If $x \prec x_{(1+\beta/4)\epsilon},$ then $\#\{X \in R: X \preccurlyeq x\} > (1+\beta/2)\epsilon \cdot |R|$ by definition, so $k_1 := \#\{X \in R: X \prec x_{(1+\beta/4)\epsilon}\} > (1+\beta/2) \epsilon \cdot |R|.$ However, $k_1 \sim \text{Bin}(|R|, \eta)$ where $\eta \le (1+\beta/4) \epsilon$ by definition. Hence, the probability that $x \prec x_{(1+\beta/4)\epsilon}$ is at most $\BP\left(\text{Bin}(|R|, (1+\beta/4)\epsilon) > (1+\beta/2) \cdot |R|\right)$, which by the Chernoff bound is at most
\[\exp\left(-((\beta/4)/(1+\beta/4))^2 \cdot |R| \cdot (1+\beta/4)\epsilon/3\right) \le \exp\left(-\beta^2 \cdot |R| \cdot \epsilon/60\right) \le 1/20,\]
    where the first inequality follows since we assumed $\beta \le 0.2$ and the last inequality follows since $|R| = \frac{180}{\beta^2 \epsilon}.$
    
    Similarly, we let $x_{(1+3\beta/4)\epsilon}$ represent the $(1+3\beta/4)\epsilon$ quantile of $\Pp$. If $x \succ x_{(1+3\beta/4)\epsilon}$, then $k_2 := \#\{X \in R: X \preccurlyeq x_{(1+3\beta/4)\epsilon}\} \le \#\{X \in R: X \prec x\} \le (1+\beta/2) \epsilon \cdot |R|.$ As $k_2 \sim \text{Bin}(|R|, \eta)$ for some $\eta > (1+3\beta/4) \epsilon$, the probability that $x \succ x_{(1+3\beta/4)\epsilon}$ is at most $\BP\left(\text{Bin}(|R|, (1+3\beta/4)\epsilon) \le (1+\beta/2) \cdot |R|\right)$, which by the Chernoff bound is at most
\[\exp\left(-((\beta/4)/(1+3\beta/4))^2 \cdot |R| \cdot (1+3\beta/4)\epsilon/3\right) \le \exp\left(-\beta^2 \cdot |R| \cdot \epsilon/60\right) \le 1/20.\]

    So, with probability at least $9/10$, $x_{(1+\beta/4)\epsilon} \preccurlyeq x \preccurlyeq x_{(1+3\beta/4)\epsilon}.$ In this case, there must exist  $\epsilon^* \in [(1+\beta/4) \epsilon, (1 + 3\beta/4) \epsilon]$ such that $x$ is the $\epsilon^*$ quantile of $\Pp$.
\end{proof}

We next prove the following auxiliary lemma.

\begin{lemma} \label{lem:quantile_ess}
For any $\epsilon < 1$, recall that $x_\epsilon$ represents the $\epsilon$ quantile of $\Pp$ and $p_\epsilon = p(x_\epsilon)$.
Then, for any $0 < \epsilon, \alpha < 1,$ it holds that $\ess_{(1-\alpha)\epsilon} \geq \frac{\epsilon \alpha}{p_\epsilon}$.
\end{lemma}
\begin{proof}
    Assume without loss of generality that the elements are sorted in increasing order of probability, i.e., $p(x_1) \le p(x_2) \le \cdots \leq p(x_n)$. We may also assume all elements have nonzero probability by removing all elements with $0$ probability. (Indeed, this does not affect $\ess_{(1-\alpha)\epsilon}$ or $p_\epsilon$.) For simplicity, we define $a := x_{(1-\alpha)\epsilon}$ and $b := x_{\epsilon}$.
    Then, for $X \sim \Pp$, $\BP(X \prec a) \le (1-\alpha)\epsilon < \BP(X \preccurlyeq a)$, and $\BP(X \prec b) \le \epsilon < \BP(X \preccurlyeq b)$. Importantly, this means $\BP(a \preccurlyeq X \preccurlyeq b) = \BP(X \preccurlyeq b) - \BP(X \prec a) > \alpha \cdot \epsilon$. However, $p(b) = p_\epsilon,$ and $p(c) \le p_\epsilon$ for any $a \preccurlyeq c \preccurlyeq b.$ Thus, $(b-a+1) \cdot p_\epsilon \ge \BP(a \preccurlyeq X \preccurlyeq b) > \alpha \cdot \epsilon$, which means that $b-a+1 > \frac{\alpha \cdot \epsilon}{p_\epsilon}$.
    %In addition, in the event that $\BP(X \prec b) = \epsilon$, then in fact $\BP(a \preccurlyeq X \prec b) = \BP(X \prec b) - \BP(X \prec a) \ge \alpha \cdot \epsilon$, which means that $b-a \ge \frac{\alpha \cdot \epsilon}{p_\epsilon}$.
    
    %Next, we show that $\ess_{(1-\alpha)\epsilon} \ge b-a+1$ if $\BP(X \prec b) > \epsilon$, and $\ess_{(1-\alpha)\epsilon} \ge b-a$ if $\BP(X \prec b) = \epsilon$.
   % Indeed, to minimize the support size of any distribution within $(1-\alpha)\epsilon$ of $\Pp$, we remove all elements with the smallest probabilities.
    %By \Cref{obs:equivalent_definition}, $\epsilon$-effective support size corresponds to the number of elements after removing the first elements with total probability mass $\epsilon$.
    %However, we cannot remove the first $a$ elements if $\BP(X \le a) > (1-\alpha) \epsilon$, so we cannot remove more than $a-1$ elements. But the number of elements with nonzero probability is at least $b$, so the effective support size $\ess_{(1-\alpha)\epsilon}$ is at least $b-(a-1) = b-a+1$. If $\BP(X \le a) = (1-\alpha) \epsilon$, then we cannot remove $a+1$ elements, so effective support size $\ess_{(1-\alpha)\epsilon}$ is at least $b-a$.
    Next, we remark that since $a$ is the $(1-\alpha)\epsilon$ quantile of $\Pp$, the $(1-\alpha)\epsilon$ effective support size is precisely the number of elements which are $\succcurlyeq a$. Since all elements between $a$ and $b$ in the order fall in this category, we have that $\ess_{(1-\alpha)\epsilon} \ge b-a+1$.
    
    %To summarize, we have that either $\ess_{(1-\alpha)\epsilon} \ge b-a+1 \ge \frac{\alpha \cdot \epsilon}{p_\epsilon}$ (if $\BP(X \le a) > (1-\alpha) \epsilon$) or $\ess_{(1-\alpha)\epsilon} \ge b-a \ge \frac{\alpha \cdot \epsilon}{p_\epsilon}$ (if $\BP(X \le a) = (1-\alpha) \epsilon$).
    To summarize, we have $\ess_{(1-\alpha)\epsilon} \ge b-a+1 \ge \frac{\alpha \cdot \epsilon}{p_\epsilon}$, which completes the proof.
\end{proof}

We are now ready to prove our main result.

\begin{theorem} \label{thm:main}
Suppose that $0 < \epsilon < 1$ and $0 < \beta, \gamma \le 0.2$. Then, with probability at least $2/3$, \Cref{alg:main} returns a $(1+ \gamma)$-factor approximation to the $[\epsilon, (1+\beta)\epsilon]$ effective support size. Its sample complexity is $O(\frac{1}{\beta^2 \epsilon} + \frac{1}{\epsilon \beta \gamma^2})$.
\end{theorem}

\begin{proof}
The sample complexity is clearly as claimed. We thus focus on correctness. Recall that we may assume WLOG that $(1+\beta) \cdot \epsilon < 1$.
We will show that $(1-0.4 \gamma) \ess_{(1+\beta)\epsilon} \le S \le (1 + 0.4 \gamma) \ess_{\epsilon},$ where $S$ is defined in Line 4 of \Cref{alg:main}. Since our final estimate is $(1 + 0.5 \gamma) S$, and since $1 \le (1-0.4 \gamma) \cdot (1+0.5 \gamma)$ and $(1+0.4 \gamma) \cdot (1+0.5 \gamma) \le 1+\gamma$ for $\gamma \le 0.2$, this implies our final estimate is in the range $[\ess_{(1+\beta)\epsilon}, (1+\gamma) \ess_{\epsilon}]$, as desired.

Recall that $x$ is the element generated in Line 2 of \Cref{alg:main}. Define $p := p(x)$, and let $\epsilon^*$ be such that $x$ is the $\epsilon^*$ quantile of $\Pp$. Let $\mathcal{E}_1$ denote the event that we can choose $\epsilon^* \in [(1+\beta/4)\epsilon, (1+3\beta/4)\epsilon]$. (By \Cref{lem:first_2_lines}, $\mathcal{E}_1$ holds with at least $9/10$ probability.)
Consider the random variable $Y = \BI[X \succcurlyeq x]/p(X)$ for $X \sim \Pp$. We have that
\[
\BE[Y|\epsilon^*] = \sum_{y \in U} p(y) \frac{\BI[y \succcurlyeq x]}{p(y)} = \#\{y \in U: y \succcurlyeq x\} = \ess_{\epsilon^*}.
\]
At the same time, note that $\sup[Y|\epsilon^*] \le \frac{1}{p}$ (where we recall $p := p(x)$ and $x = x_{\epsilon^*}$) since $\BI[X \succcurlyeq x] = 1$ implies $p(X) \ge p$. Therefore, 
\[
\Var[Y|\epsilon^*] \leq \BE[Y^2|\epsilon^*] \leq \BE[Y|\epsilon^*] \sup[Y|\epsilon^*] \leq \ess_{\epsilon^*}/p.
\]
We thus have
\[
\BE[S|\epsilon^*] = \BE[Y|\epsilon^*] = \ess_{\epsilon^*},
\]
recalling that $S$ is an average of $t$ copies of the random variable $Y$.
%, where the second equality uses that $S$ is independent of $t$ (as it is also independent of $R$ and thus on $\epsilon^*$). 
It also holds that
\[
\Var[S|\epsilon^*] = \Var[Y|\epsilon^*]/t \leq \ess_{\epsilon^*}/(t p),
\]
where we note that $t = \frac{500}{\epsilon \beta \gamma^2}$ depends on $\epsilon$ but not on $\epsilon^*$. Conditioning on $\epsilon^* \geq \epsilon$ (equivalently, $t \geq \frac{500}{\epsilon^* \beta \gamma^2}$, which holds on $\mathcal{E}_1$), we have that
\[
\Var[S|\epsilon^*, \epsilon^* \ge \epsilon] \leq \epsilon^* \beta \gamma^2 \ess_{\epsilon^*}/(500p).
\]
%By the Fubini's theorem, integrating over all $t'$ that correspond to $\epsilon^*$ such that $(1+\beta/4)\epsilon \leq \epsilon^* \leq (1+3\beta/4)\epsilon$ (note that this exactly corresponds to integrating over $\mathcal{E}_1$, by the definition of $\mathcal{E}_1$) we have $\BE[S | \mathcal{E}_1] = \ess_{\epsilon^*}$. Moreover, since $t \geq \frac{20}{\epsilon^* \beta \gamma^2}$ assuming $\mathcal{E}_1$, we have
%\[
%\Var[S|\mathcal{E}_1] \leq \epsilon^* \beta \gamma^2 \ess_{\epsilon^*}/(20p)
%\]
Therefore, by the (conditional) Chebyshev inequality, assuming $\mathcal{E}_1$ and conditioning on $\epsilon^* \in [(1+\beta/4)\epsilon, (1+3\beta/4) \epsilon]$, it holds with probability at least $9/10$ that
%\begin{itemize}
%\item show that $\BE[S] = \ess_{\epsilon^*}$
%\item the variance of $S$ is $\leq E*Sup/t \leq \ess_{\epsilon^*}/(t p)$
%\item use chebyshev to give bounds on deviation of $S$
%
%By the chebyshev inequailty, we have that, with probability $9/10$, it holds
\[
\left|S - \ess_{\epsilon^*}\right| \le \sqrt{\epsilon^*\cdot \beta \cdot \gamma^2/50} \cdot \sqrt{\ess_{\epsilon^*}/p}.
\]
We call the event when this is the case $\mathcal{E}_2$.

We split the rest of the analysis into two main cases. The first case is when $\ess_{\epsilon^*} \ge \frac{\epsilon^* \beta}{8 p},$ and the second case is when $\ess_{\epsilon^*} \le \frac{\epsilon^* \beta}{8 p}.$

We start by the simple case when $\ess_{\epsilon^*} \geq \frac{\epsilon^* \beta}{8 p}$. In this case, the value $\ess_{\epsilon^*}$ is relatively large, and this already guarantees a good approximation. %Specifically, assume that $\ess_{\epsilon^*} \geq \frac{\epsilon^* \beta}{8 p}$.
%We now consider the case $\ess_{\epsilon^*} \geq \frac{\epsilon^* \beta}{8 p}$.
Since $\ess_{\epsilon^*} \geq \frac{\epsilon^* \beta}{8 p}$, we have that $1/p \leq \frac{8\ess_{\epsilon^*}}{\epsilon^* \beta}$.
It therefore holds on $\mathcal{E}_2$ that
\[
\left|S - \ess_{\epsilon^*}\right| \le \sqrt{\epsilon^*\cdot \beta \cdot \gamma^2/50} \cdot \sqrt{\ess_{\epsilon^*}/p} \le 0.4 \gamma \cdot \ess_{\epsilon^*}.
\]

We now consider the case when $\ess_{\epsilon^*} \leq \frac{\epsilon^* \beta}{8 p}$. %; we consider the complement case below.
%It holds $S' \geq \frac{\epsilon* \beta}{4 p}$, and we thus return $S \geq \ess_{\epsilon^*}$. 
It holds, by \Cref{lem:quantile_ess}, that $\ess_{(1-\beta/4)\epsilon^*} \geq \frac{\epsilon^*\beta}{4 p}$, and it therefore holds $\ess_{\epsilon^*} \leq \ess_{(1-\beta/4)\epsilon^*}/2$. We may thus bound
\begin{align*}
\ess_{\epsilon^*} + \sqrt{\epsilon^* \cdot \beta \cdot \gamma^2/50} \cdot \sqrt{\ess_{\epsilon^*}/p} \leq& \ess_{\epsilon^*} + \sqrt{\epsilon^* \beta/8} \cdot \sqrt{\ess_{\epsilon^*}/p} \\
\leq&  \ess_{\epsilon^*} + \sqrt{\ess_{\epsilon^*} \ess_{(1-\beta/4)\epsilon^*}/2} \\ \leq& \ess_{(1-\beta/4)\epsilon^*}/2 + \sqrt{\ess_{(1-\beta/4)\epsilon^*} \cdot \ess_{(1-\beta/4)\epsilon^*}/4} \\=& \ess_{(1-\beta/4) \epsilon^*} \leq \ess_\epsilon
\end{align*}
where the last inequality holds on $\mathcal{E}_1$.
%Moreover, it holds $\frac{\epsilon \beta}{6 p} \leq \ess_{(1-\beta/4)\epsilon}$, so it also holds $S' \leq \ess_{(1-\beta/4)\epsilon}$

Next, we need to argue that $S \geq (1-0.4 \gamma) \ess_{(1+\beta)\epsilon}$. We do this by defining a random variable $S'$ that is stochastically dominated by $S$, and at the same time it has low enough variance that we may use the Chebyshev inequality to show that, with high constant probability, $S' \geq \ess_{(1+\beta)\epsilon}$. 
Specifically, we define $S' = \frac{1}{t} \sum_{i=1}^t Y_i'$, where
% for $Y_i'$ being $1/p(y)$ for $y$ such that $\BP(X \succcurlyeq y) \geq 1-(1+\beta)\epsilon$ for $X \sim \Pp$, as $1/p_{(1+\beta)\epsilon}$ for $1-(1+\beta)\epsilon > \BP[X \succcurlyeq y] \geq 1-\epsilon^*$ and $0$ otherwise.
\[Y_i' := \begin{cases}
    1/p(y_i) & y_i \succcurlyeq x_{(1+\beta)\epsilon}\\
    1/p_{(1+\beta)\epsilon} & x \preccurlyeq y_i \prec x_{(1+\beta) \epsilon} \\
    0 & y_i \prec x,
\end{cases}\]
where we recall that each $y_i \overset{i.i.d.}{\sim} \Pp$. (Recall that $x_{(1+\beta)\epsilon}$ is the $(1+\beta)\epsilon$ quantile of $\Pp$, and $x$ is the $\epsilon^*$ quantile of $\Pp$.) Note that this also implies each $Y_i'$ is i.i.d.

We now prove that $S'$ is stochastically dominated by $S$. The random variable $S$ is average of $t$ independent copies of $Y$ while $S'$ is an average of $t$ random variables $Y'_i$. It is thus sufficient to prove that $Y$ stochastically dominates $Y'_i$, since the $Y'_i$ variables are i.i.d. We do this by demonstrating a coupling between $Y$ and $Y_i'$ in which it always holds $Y \geq Y'_i$. Specifically, consider $Y$ and $Y'_i$ with the same sample $y$.
%If $\BP(X \succcurlyeq y) \geq 1-(1+\beta)\epsilon$, then $Y = Y'_i = 1/p(y)$. If $1-(1+\beta)\epsilon > \BP[X \succcurlyeq y] \geq 1-\epsilon^*$, then $Y'_i = 1/p_{(1+\beta)\epsilon}$ while $Y = 1/p(y) \geq 1/p_{(1+\beta)\epsilon}$ where the inequality holds because the condition $1-(1+\beta)\epsilon > \BP[X \succcurlyeq y] \geq 1-\epsilon^*$ says that $y$ is between $\epsilon^*$ and $(1+\beta)\epsilon$ quantitle with respect to $\prec$; since prec orders elements by their probabilities $p(\cdot)$, this means $p(y)$ is lower than the probability of the $(1+\beta)\epsilon$ quantitle, or in other words, $1/p(y) \geq 1/p_{(1+\beta)\epsilon}$\jakub{This sentence was written in haste. Please check.}. Otherwise $Y = Y'_i = 0$.
We have the following three cases.
\begin{enumerate}
    \item If $y \succcurlyeq x_{(1+\beta)\epsilon}$, then $Y_i' = 1/p(y) = Y$, since the indicator of $y \succcurlyeq x$ is $1$.
    \item If $x \preccurlyeq y \prec x_{(1+\beta)\epsilon}$, then $Y = 1/p(y)$ and $Y_i' = 1/p_{(1+\beta)\epsilon}$. However, $p(y) \le p(x_{(1+\beta)\epsilon}) = p_{(1+\beta)\epsilon}$, so $1/p(y) \ge 1/p_{(1+\beta)\epsilon}$.
    \item If $y \prec x$, then $Y = Y_i' = 0$, where $Y = 0$ since the indicator of $y \succcurlyeq x$ is $0$.
\end{enumerate}
In all cases, $Y \ge Y_i'$, so $Y$ stochastically dominates $Y'_i$. Thus, $S$ stochastically dominates $S'$.

At the same time, assuming $\epsilon^*$ is such that $x = x_{\epsilon^*} \preccurlyeq x_{(1+\beta)\epsilon}$, it holds that
%\begin{align*}
%\BE[Y_i'] =& \BE[\BI[\BP[X \succcurlyeq y|y] \geq 1-(1+\beta)\epsilon]/p(y)] \\&+ \BE[\BI[1-(1+\beta)\epsilon > \BP[X \succcurlyeq y|y] \geq 1-\epsilon^*]/p_{(1+\beta)\epsilon}] \\
%\leq& \ess_{(1+\beta)\epsilon} + \beta \epsilon/(p_{(1+\beta)\epsilon} 4)
%\end{align*}
\begin{align*}
    \BE[Y_i'|\epsilon^*] &= \sum_{\substack{y \in U \\ y \succcurlyeq x_{(1+\beta)\epsilon}}} p(y) \cdot \frac{1}{p(y)} + \sum_{\substack{y \in U \\ x \preccurlyeq y \prec x_{(1+\beta)\epsilon}}} p(y) \cdot \frac{1}{p_{(1+\beta)\epsilon}} \\
    &= \#\{y \in U: y \succeq x_{(1+\beta)\epsilon}\} + \frac{1}{p_{(1+\beta)\epsilon}} \cdot \BP_{X \sim \Pp}(x \preccurlyeq X \prec x_{(1+\beta)\epsilon}) \\
    &= \ess_{(1+\beta)\epsilon} + \frac{1}{p_{(1+\beta)\epsilon}} \cdot \left(\BP_{X \sim \Pp}(X \prec x_{(1+\beta)\epsilon})-\BP_{X \sim \Pp}(X \prec x)\right).
\end{align*}
    Since $x$ is the $\epsilon^*$ quantile, $\BP(X \prec x) \le \epsilon^*$, and since $x_{(1+\beta)\epsilon}$ is the $(1+\beta)\epsilon$ quantile, $\BP(X \preccurlyeq x_{(1+\beta)\epsilon}) > (1+\beta)\epsilon$, which means $\BP(X \prec x_{(1+\beta)\epsilon}) > (1+\beta)\epsilon - \BP(X = x_{(1+\beta)\epsilon}) = (1+\beta)\epsilon - p_{(1+\beta)\epsilon}$. In addition, we also have that $\BP(X \prec x_{(1+\beta)\epsilon}) - \BP(X \prec x) \ge 0$. Therefore, we have
\begin{align*}
    \BE[Y_i'|\epsilon^*] &\ge \ess_{(1+\beta)\epsilon} + \frac{1}{p_{(1+\beta)\epsilon}}\max\left((1+\beta)\epsilon - p_{(1+\beta)\epsilon} - \epsilon^*, 0\right) \\
    &= \ess_{(1+\beta)\epsilon} + \max\left(\frac{\beta \epsilon}{4 p_{(1+\beta)\epsilon}} - 1, 0\right),
\end{align*}
where the last inequality holds on $\mathcal{E}_1$, 
%since on that event $\BP[1-(1+\beta)\epsilon > \BP[X \succcurlyeq y|y] \geq 1-\epsilon^*] = (1+\beta)\epsilon - \epsilon^* \geq \beta \epsilon/4$. Again, we are here assuming $X \sim \Pp$. Note that $\BP[X \succcurlyeq y|y]$ is a random variable, depending on $y$.
since that implies $\epsilon^* \le (1+3\beta/4)\epsilon$.
Since $\ess_{(1+\beta)\epsilon} \ge 1,$ this means that assuming $\mathcal{E}_1$, 
\[\BE[Y_i'|\epsilon^*] \ge \max\left(\ess_{(1+\beta)\epsilon}, \frac{\beta \epsilon}{4 p_{(1+\beta)\epsilon}}\right).\]

Next, assuming $\mathcal{E}_1$, we have that
\[
\Var[Y_i'|\epsilon^*] \leq \BE[Y_i'|\epsilon^*] \sup[Y_i'|\epsilon^*] = \BE[Y_i'|\epsilon^*]/p_{(1+\beta)\epsilon} \leq 4 \BE[Y_i'|\epsilon^*]^2/(\beta \epsilon),
\]
where the last inequality holds because $\BE[Y'|\epsilon^*] \geq \beta \epsilon/(4 p_{(1+\beta)\epsilon})$. Therefore, 
\[
\Var[S'|\epsilon^*] \leq 4 \BE[Y_i'|\epsilon^*]^2/(\beta \epsilon t) \leq \gamma^2 \BE[Y_i'|\epsilon^*]^2/125 = \gamma^2 \BE[S'|\epsilon^*]^2/125.
\]
By the Chebyshev inequality, we then have with probability at least $9/10$, that $S' \geq (1-0.4\gamma) \BE[S'|\epsilon^*] = (1-0.4\gamma) \BE[Y_i'|\epsilon^*] \geq (1-0.4\gamma) \ess_{(1+\beta)\epsilon}$. We call the event when this happens $\mathcal{E}_3$.

We have shown that the algorithm gives a correct output on $\mathcal{E}_1 \cap \mathcal{E}_2 \cap \mathcal{E}_3$. It holds that
\begin{align*}
\BP[\mathcal{E}_1 \cap \mathcal{E}_2 \cap \mathcal{E}_3] =& 1- \BP[\neg \mathcal{E}_1 \cup \neg \mathcal{E}_2 \cup \neg \mathcal{E}_3] \\
\geq& 1-\BP[\neg \mathcal{E}_1] - \BP[\neg \mathcal{E}_2 | \mathcal{E}_1] - \BP[\neg \mathcal{E}_3 | \mathcal{E}_1] > 2/3,
\end{align*}
%\jakub{maybe add one more line to the math for more detail}
where the last inequality holds because we bounded above each of the three probabilities by $1/10$.
\end{proof}

As a direct corollary of combining Theorem~\ref{thm:main} with Observation~\ref{obs:tradeoff_error_parameters}, we have the following.

\begin{corollary} \label{cor:main}
    By setting $\gamma = \epsilon \cdot \beta$ in \Cref{alg:main} and outputting $(1+\gamma/2)/(1+\gamma) \cdot S$ instead of $(1+\gamma/2) \cdot S$ in the final line of \Cref{alg:main}, we return an $[\epsilon, (1+2\beta)\epsilon]$-effective support size. The sample complexity is $O(\frac{1}{\epsilon^3 \beta^3})$.
\end{corollary}

\section*{Acknowledgments}

%Jakub T\v{e}tek is supported by the VILLUM Foundation grant 16582.
%Shyam Narayanan is supported by the NSF GRFP Fellowship, and the NSF TRIPODS Program (award DMS-2022448).
Shyam Narayanan would like to thank Prof. Mikkel Thorup and BARC for allowing him to visit BARC, which facilitated this research.

%\hrule
%
%%\begin{algorithm}
%%$X \leftarrow$ sample $t = \frac{\ctodo}{\epsilon \gamma^2} + \frac{\ctodo}{\epsilon \beta^2}$ items\\
%%$Y \leftarrow$ sort $X$\\
%%\Return{$\frac{1}{t} \sum_{i=\lfloor (1+\beta/2)\epsilon t \rfloor}^t \frac{1}{p(Y_i)}$}
%%\end{algorithm}
%
%\begin{proof}
%Let $p = P(Y_{\lfloor (1+\beta/2)\epsilon t \rfloor})$.
%Let $\epsilon^* = P(x \prec Y_{\lfloor (1+\beta/2)\epsilon t \rfloor})$.
%\end{proof}

\bibliographystyle{plainnat}
\bibliography{literature}
\end{document}